%% file: main.tex
\documentclass{article}

\input{__arxiv_preamble}

\usepackage{microtype}
  \usepackage{amsmath,amsfonts}
  \usepackage{tikz}
  \usepackage[cmbtt]{bold-extra}
  \usetikzlibrary{decorations.pathmorphing}
  \usetikzlibrary{decorations.markings}
  \usetikzlibrary{decorations.pathmorphing,shapes}
  \usetikzlibrary{calc,decorations.pathmorphing,shapes}
  \usetikzlibrary{snakes}
  \usepackage{tikz-qtree}
  \usepackage[T1]{fontenc}
  \usepackage[utf8]{inputenc}
  \usepackage{thm-restate}
  \usepackage{comment}
  \usepackage{forest}
  \usepackage{xspace}
  \usepackage{enumerate}
  \usepackage{makecell}
  \usepackage{listings}
  \usepackage{multirow}
  \usepackage{diagbox}
  \usepackage{todonotes}
  \usepackage{thmtools}
  \usepackage{thm-restate}
    \usepackage[ruled,noline,noend]{algorithm2e}
  \usepackage[capitalise]{cleveref}
  \usepackage{graphics,adjustbox}
  \usepackage{tabularx}
  \usepackage{amssymb}
  \usepackage{epsfig}
  \usepackage{verbatim}

\def\dd{\mathinner{.\,.}}
\newcommand{\Oh}{\mathcal{O}}
\newcommand{\kLCFFull}{\textsc{Longest Common Factor with at most $k$ Mismatches}\xspace}
\newcommand{\kLLCFFull}{\textsc{LCF of Length at Least $\ell$ with at most $k$ Mismatches}\xspace}
\newcommand{\kLCF}{\textsc{LCF}$_k$}
\newcommand{\oLCF}{\textsc{LCF}$_1$}

\newcommand{\Pairs}{\mathsf{Pairs}}
\newcommand{\Sh}{\Sigma_{\$}}

\newcommand{\F}{\mathcal{F}}

\newcommand{\HD}{d_H}
\renewcommand{\L}{\mathcal{L}}
\newcommand{\sub}{\subseteq}

\newcommand{\LCF}{\mathsf{LCF}}
\newcommand{\LCP}{\mathsf{LCP}}
\newcommand{\LCS}{\mathsf{LCS}}
\newcommand{\eps}{\varepsilon}

\newcommand{\floor}[1]{\lfloor #1 \rfloor}

\newcommand{\T}{\mathcal{T}}

 \newcommand{\defproblem}[3]{
  \vspace{2mm}
\noindent\fbox{
  \begin{minipage}{0.96\textwidth}
  #1\\
  {\bf{Input:}} #2  \\
  {\bf{Output:}} #3
  \end{minipage}
  }
  \vspace{2mm}
}

\renewcommand{\S}{\mathbf{S}}

\newcommand{\TT}{\mathcal{T}}

\newtheorem{prop}[theorem]{Proposition}

\begin{document}
  \maketitle

\begin{abstract}
In the Longest Common Factor with $k$ Mismatches (\kLCF) problem, we are given two strings $X$ and $Y$ of total length $n$,
and we are asked to find a pair of maximal-length factors, one of $X$ and the other of $Y$, such that their Hamming distance is at most $k$. 
Thankachan et al.~\cite{DBLP:journals/jcb/ThankachanAA16} show that this problem can be solved in $\Oh(n \log^k n)$ time and $\Oh(n)$ space for constant $k$. 
We consider the \kLCF($\ell$) problem in which we assume that the sought factors have length at least $\ell$, and the \kLCF($\ell$) problem for $\ell=\Omega(\log^{2k+2} n)$, which we call the Long \kLCF\ problem.
We use difference covers to reduce the Long \kLCF\ problem to a task involving $m=\Oh(n/\log^{k+1}n)$ \emph{synchronized} factors. 
The latter can be solved in $\Oh(m \log^{k+1}m)$ time, which results in a linear-time algorithm for Long \kLCF.
In general, our solution to \kLCF($\ell$) for arbitrary $\ell$ takes $\Oh(n + n \log^{k+1} n/\sqrt{\ell})$ time.
\end{abstract}

\section{Introduction}
The longest common factor (LCF) problem is a classical and well-studied problem in theoretical computer science.
It consists in finding a maximal-length factor of a string $X$ occurring in another string $Y$.
When $X$ and $Y$ are over a linearly-sortable alphabet, the LCF problem can be solved in the optimal $\Oh(n)$ time and space~\cite{DBLP:conf/cpm/Hui92,DBLP:books/cu/Gusfield1997}, where $n$ is the total length of $X$ and $Y$. 
Considerable efforts have thus been made on improving the {\em additional} working space; namely, the space required for computations, not taking into account the space providing read-only access to $X$ and $Y$. 
We refer the interested reader to~\cite{DBLP:conf/cpm/StarikovskayaV13,DBLP:conf/esa/KociumakaSV14}.

In many bioinformatics applications and elsewhere, it is relevant to consider potential alterations within the pair of input strings (e.g.~DNA sequences). It is thus natural to define the LCF problem under a distance metric model. 
The problem then consists in finding a pair of maximal-length factors of $X$ and $Y$ whose  distance is at most $k$. In fact, this problem has received much attention recently, in particular due to its applications in alignment-free sequence comparison~\cite{DBLP:journals/jcb/UlitskyBTC06,kmacs}.

Under the Hamming distance model, the problem is known as the \kLCFFull (\kLCF) problem.
The restricted case of $k=1$ was first considered in~\cite{DBLP:journals/poit/BabenkoS11}, where an $\Oh(n^2)$-time and $\Oh(n)$-space solution was given. It was later improved by Flouri et al.~\cite{DBLP:journals/ipl/FlouriGKU15}, who built heavily on a technique by Crochemore et al.~\cite{DBLP:journals/tcs/CrochemoreIMS06} to obtain $\Oh(n \log n)$ time and $\Oh(n)$ space.

For a general value of $k$, the problem can be solved in $\Oh(n^2)$ time and space by a dynamic programming algorithm, but more efficient solutions have been devised.
Leimeister and Morgenstern~\cite{kmacs} first suggested a greedy heuristic algorithm. Flouri et al.~\cite{DBLP:journals/ipl/FlouriGKU15} proposed an $\Oh(n^2)$-time algorithm that uses $\Oh(1)$ additional space. 
Grabowski~\cite{DBLP:journals/ipl/Grabowski15} presented two algorithms with running times $\Oh (n ((k+1) (\ell_0+1))^k)$ and $\Oh (n^2 k/\ell_k)$, where $\ell_0$ and $\ell_k$ are, respectively, the length of an LCF of $X$ and $Y$ and the length of an \kLCF{} of $X$ and $Y$. 
Thankachan et al.~\cite{DBLP:journals/jcb/ThankachanAA16} proposed an $\Oh(n \log^k n)$-time and $\Oh(n)$-space algorithm (for any constant $k$).

Abboud et al.~\cite{Abboud:2015:MAP:2722129.2722146} employed the polynomial method to obtain a $k^{1.5} n^2 / 2^{\Omega(\sqrt{\frac{\log n}{k}})}$-time randomized algorithm.
 Kociumaka et al.~\cite{DBLP:journals/corr/abs-1712-08573} showed that a strongly subquadratic-time algorithm for the \kLCF{} problem, for binary strings and $k=
\Omega(\log n)$, refutes the Strong Exponential Time Hypothesis~\cite{DBLP:journals/jcss/ImpagliazzoPZ01,DBLP:journals/jcss/ImpagliazzoP01}. Thus, subquadratic-time solutions for approximate variants of the problem have been developed \cite{DBLP:journals/corr/abs-1712-08573,DBLP:conf/cpm/Starikovskaya16}. The average-case complexity of this problem has also been considered~\cite{DBLP:journals/jcb/ThankachanCLAA16,DBLP:conf/sofsem/AlamroACIP18,PLCP}.

\subsection{Our Contribution}
We consider the following variant of the \kLCFFull problem in which the result is constrained to have at least a given length. Let \kLCF$(X,Y)$ denote the length of the longest common factor of $X$ and $Y$ with at most $k$ mismatches.

\defproblem{\kLLCFFull (\kLCF$(X,Y,\ell)$)}{Two strings $X$ and $Y$ of total length $n$ and integers $k\ge 0$ and $\ell\ge 1$}{\kLCF$(X,Y)$ if it is at least $\ell$, and ``NONE'' otherwise.}

We focus on a special case of this problem with $\ell=\Omega(\log^{2k+2} n)$ which we call \textsc{Long} \kLCF\ problem.

Apart from its theoretical interest, solutions to the \kLCF$(X,Y,\ell)$ problem, and \textsc{Long} \kLCF\ in particular, may prove to be useful from a practical standpoint. The \kLCF{} length has been used as a measure of sequence similarity~\cite{DBLP:journals/jcb/UlitskyBTC06,kmacs}. It is thus assumed that similar sequences share relatively long factors with $k$ mismatches.

We show an $\Oh(n)$-time algorithm for the \textsc{Long} \kLCF\ problem. 
Moreover, we prove that \kLCF$(X,Y,\ell)$ can be solved in $\Oh(n + n \log^{k+1} n/\sqrt{\ell})$ time for arbitrary $\ell$
and constant $k$.
In the final section we discuss the complexity for $k=\Oh(\log n)$.
This unveils that the $\Oh(\cdot)$ notation hides a multiplicative factor that is actually subconstant in $k$.

For simplicity, we only describe how to compute the length $\LCF_k(X,Y)$. 
It is straightforward to amend our solution so that it extracts the corresponding factors of $X$ and~$Y$.

\subparagraph*{\bf Toolbox.}
We use the following algorithmic tools:
\begin{itemize}
  \item Difference covers (see, e.g., \cite{DBLP:journals/tocs/Maekawa85,BurkhardtEtAl2003}) let us reduce the \kLCF$(X,Y,\ell)$ problem to searching for longest common prefixes and suffixes with at most $k$ errors ($\LCP_k$, $\LCS_k$) at positions belonging to sets $A$ in $X$ and $B$ in $Y$ such that $|A|,|B| = \Oh(n/\sqrt{\ell})$.
  \item We use a technique of recursive heavy-path decompositions by Cole et al.~\cite{DBLP:conf/stoc/ColeGL04}, already used in the context of the \kLCF\ problem by Thankachan et al.~\cite{DBLP:journals/jcb/ThankachanAA16}, to reduce computing $\LCP_k$, $\LCS_k$ to computing $\LCP$, $\LCS$ in sets of modified prefixes and suffixes starting at positions in $A$ and $B$. Modifications consist in at most $k$ changes and increase the size of the problem by a factor of $\Oh(\log^k n)$. We adjust the original technique of Cole et al.~\cite{DBLP:conf/stoc/ColeGL04} so that all modified strings are stored in one compacted trie. Details are given in the appendix. 
  \item Finally we apply to the compacted trie a solution to a problem on colored trees that is the cornerstone of the previous $\Oh(n \log n)$-time solution for \oLCF\ problem by Flouri et al.~\cite{DBLP:journals/ipl/FlouriGKU15} (and originates from efficient merging of AVL trees~\cite{DBLP:journals/jacm/BrownT79}). 
\end{itemize}
In total we arrive at $\Oh(n\log^{k+1} n / \sqrt{\ell}+n)$ complexity.

\section{Preliminaries}
Henceforth we denote the input strings by $X$ and $Y$ and their common length by~$n$.

The $i$-th letter of a string $U$, for $1 \le i \le |U|$, is denoted by $U[i]$.
By $[i \dd j]$ we denote the integer interval $\{i,\ldots,j\}$ and by $U[i \dd j]$ we denote the string $U[i] \ldots U[j]$ that we call a factor of $U$.
For simplicity, we denote $U[\dd i]=U[1 \dd i]$ and $U[i \dd]=U[i \dd |U|]$.
By $U^R$ we denote the mirror image of $U$.

For a pair of strings $U$ and $V$ such that $|U|=|V|$, we define their Hamming distance as $\HD(U,V)=|\{1 \le i \le |U|\,:\, U[i] \ne V[i]\}|$.
For two strings $U,V$ and a non-negative integer $d$,
we define
$$\LCP_d(U,V)\ =\ \max\{p\le |U|,|V|\,:\,\HD(U[1 \dd p],V[1 \dd p])\le d\}.$$

Let $T$ be the trie of a collection of strings $\F$. 
The \emph{compacted trie} of $\F$, $\TT(\F)$, contains the root, the branching nodes, and the terminal nodes of $T$. 
 Each edge of the compacted trie may represent several edges of $\TT$ and is labeled by a factor of one of the strings $F_i$, stored in $\Oh(1)$ space. 
 The edges outgoing from a node are labeled by the first letter of the respective strings. 
 The size of a compacted trie is $\Oh(m)$. The best-known example of a compacted trie is the suffix tree of a string; see \cite{AlgorithmsOnStrings}.

\subsection{Difference covers}

We say that a set $\S(d)\sub \mathbb{Z}_+$  is a $d$-\emph{cover} if
there is a constant-time computable function $h$ such that for $i,j\in \mathbb{Z}_+$ we have $0\le h(i,j)< d$ and $i+h(i,j), j+h(i,j)\in \S(d)$ (see Figure~\ref{fig:diff_cover_example}).
The following fact synthesizes a well-known construction implicitly used in~\cite{BurkhardtEtAl2003}, for example.
\begin{fact}[\cite{DBLP:journals/tocs/Maekawa85,BurkhardtEtAl2003}]
For each $d\in \mathbb{Z}_+$ there is a $d$-cover $\S(d)$
such that $\S_n(d):= \S(d)\cap [1\dd n]$ is of size $\Oh(\frac{n}{\sqrt{d}})$ and can be constructed in $\Oh(\frac{n}{\sqrt{d}})$ time.
\end{fact}

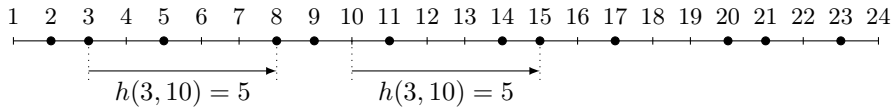
\begin{figure}[ht]

\input{_fig_diff_cover_example.tex}

\caption{An example of a 6-cover $\S_{24}(6)=\{2,3,5, 8,9,11, 14,15,19, 20,21,23\}$, with the
elements marked as black circles. For example, we may have $h(3,10)=5$ since $3+5,\,10+5\in \S_{24}(6)$. }
\label{fig:diff_cover_example}
\end{figure}

\subsection{Colored Trees Problem}
As a component of our solution we use the following problem for colored trees:

\defproblem{\textsc{Colored Trees Problem}}{Two trees $T_1$ and $T_2$ containing blue and red leaves such that each internal node is branching (except for, possibly, the root). Each leaf has a number between 1 and $m$. Each tree has at most one read leaf and at most one blue leaf with a given number. The nodes of $T_1$ and $T_2$ are weighted such that children are at least as heavy as their parent.}{A node $v_1$ of $T_1$ and a node $v_2$ of $T_2$ with maximum total weight such that $v_1$ and $v_2$ have at least one blue leaf of the same number and at least one red leaf of the same number in their subtrees.}

This abstract problem lies at the heart of the algorithm of Flouri et al.~\cite{DBLP:journals/ipl/FlouriGKU15}
for the \textsc{Longest Common Factor with 1 Mismatch} problem. They solve it in
$\Oh(m\log m)$ time applying a solution inspired by an algorithm of Crochemore et al.~\cite{DBLP:journals/tcs/CrochemoreIMS06}
finding the longest repeat with a block of $k$ don't cares,
which, in turn, is based on the fact that two AVL trees can be merged efficiently~\cite{DBLP:journals/jacm/BrownT79}.

\begin{fact}[\cite{DBLP:journals/tcs/CrochemoreIMS06,DBLP:journals/ipl/FlouriGKU15}]\label{fct:colored}
  \textsc{Colored Trees Problem} can be solved in $\Oh(m \log m)$ time.
\end{fact}

\begin{figure}[h!]
\centering
\input{_fig_tree_problem_ex1.tex}
\caption{Example instance for \textsc{Colored Trees Problem}. Assuming that each node has weight equal to the distance from the root, the optimal solution is a pair of nodes $(v_1,v_2)$ as shown in the figure. Both $v_1$ and $v_2$ have as a descendant a blue leaf with number $4$ and a red leaf with number $2$.
}
\label{fig:tree_problem_ex1}
\end{figure}
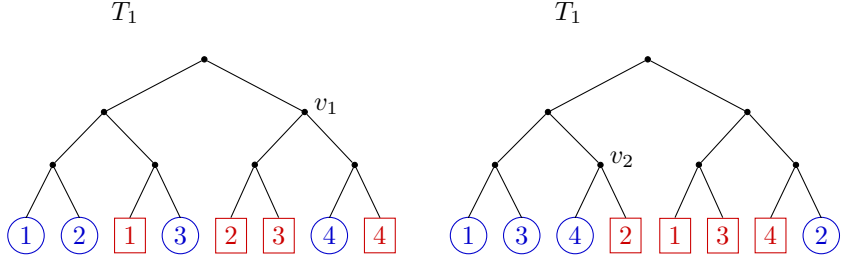

In our solution we actually use the following problem related to families of strings represented on a compacted trie. It reduces to the \textsc{Colored Trees Problem}.

\renewcommand{\problem}{\textsc{Two String Families LCP Problem}\xspace}
\renewcommand{\P}{\mathcal{P}}
\newcommand{\Q}{\mathcal{Q}}

\newcommand{\maxPairLCP}{\mathrm{maxPairLCP}}

\defproblem{\problem}{
A compacted trie $\T(\F)$ of a family of strings $\F$
and two sets $\P,\Q\sub \F^2$}
 {The value $\maxPairLCP(\P,\Q)$, defined as\\
{$\maxPairLCP(\P,\Q)\!=\!\max\{\LCP(P_1,Q_1)+\LCP(P_2,Q_2) : (P_1,P_2)\in \P \text{ and }(Q_1,Q_2)\in \Q\}$}%
 }

\begin{lemma}\label{lem:problem}
The \problem can be solved in $\Oh(|\F|+N\log N)$ time, where $N=|\P|+|\Q|$.
\end{lemma}
\begin{proof}
First, we create two copies $\T_1$ and $\T_2$ of the tree $\T(\F)$,
removing the edge labels but preserving the node weights $w(v)$ equal to the sum of lengths of edges on the path to the root.

Next, for each $(P_1,P_2)\in \P$ we attach a blue leaf to the terminal node of $\T_1$ representing $P_1$ and to the terminal of $\T_2$ representing $P_2$.
We label these two blue leaves with a unique label, denoted here $L_\P(P_1,P_2)$.
Similarly, for each $(Q_1,Q_2)\in \Q$, we attach red leaves to the terminal node of $\T_1$ representing $Q_1$
and the terminal node of $\T_2$ representing $Q_2$.
We label these two red leaves with a unique label $L_\Q(Q_1,Q_2)$.
Finally, in both $\T_1$ and $\T_2$ we remove all nodes which do not contain any colored leaf in their subtrees
and dissolve all nodes with exactly one child
(except for the roots).
This way, each tree $\T_i$ contains $\Oh(|\P|+|\Q|)$ nodes, including $|\P|+|\Q|$ leaves, each with a distinct label.

Observe that for $(P_1,P_2)\in \P$, $(Q_1,Q_2)\in \Q$, and $j\in \{1,2\}$,
the value $\LCP(P_j,Q_j)$ is the weight of the lowest common ancestor in $\T_j$
of the two leaves with labels $L_\P(P_1,P_2)$ and $L_\Q(Q_1,Q_2)$.
Consequently, our task can be formulated as follows:
Find a pair of internal nodes $v_1\in \T_1$ and $v_2\in \T_2$
of maximal total weight $w(v_1)+w(v_2)$ so that the subtrees rooted
at $v_1$ and $v_2$ contain blue leaves with the same label and red leaves with the same label.
This is exactly the \textsc{Colored Trees Problem} that can be solved in $\Oh(m\log m)$ time, where $m=|\T_1|+|\T_2|=\Oh(|\P|+|\Q|)$ (Fact~\ref{fct:colored}).
\end{proof}

\section{Reduction of \kLCF($\ell$) problem to multiple synchronized $\LCP_k$'s}
\newcommand{\LF}[1]{\overrightarrow{#1}}
\newcommand{\RF}[1]{\overleftarrow{#1}}
\newcommand{\SIM}[1]{\stackrel{#1}{\rightarrow}}
\newcommand{\pos}{\mathit{pos}}

Let $U$ be a string of length $n$.
We denote:
  $$\Pairs_\ell(U) = \{((U[\dd i-1])^R, U[i \dd])\,:\, i \in \S_n(\ell)\}.$$
Observe that 
$|\Pairs_\ell(U)|=|\S_n(\ell)|=\Oh(n/\sqrt{\ell})$.

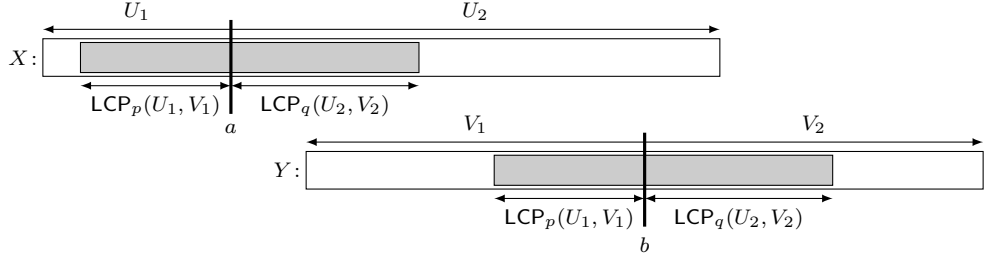
\begin{figure}[htpb]
\begin{center}
\begin{tikzpicture}[scale=0.5,every node/.style={font={\footnotesize}}]
\begin{scope}
\draw (0,0) rectangle (18,1);
\draw[fill=black!20] (1,0.1) rectangle (10, 0.9);
\draw (0,.5) node[left]{$X\!:\!$};
\draw[very thick] (5,1.5) -- (5,-1) node[below]{$a$};
\draw[-latex] (5,1.25) -- node[above]{$U_1$} (0,1.25);
\draw[-latex] (5,1.25) -- node[above]{$U_2$} (18,1.25);
\draw[latex-latex] (5,-.25) --node[below]{$\LCP_p(U_1,V_1)$} (1,-.25);
\draw[latex-latex] (5,-.25) --node[below]{$\LCP_q(U_2,V_2)$} (10,-.25);
\end{scope}

\begin{scope}[yshift=-3cm,xshift=7cm]
\draw (0,0) rectangle (18,1);
\draw[fill=black!20] (5,0.1) rectangle (14, 0.9);
\draw (0,.5) node[left]{$Y\!:\!$};
\draw[very thick] (9,1.5) -- (9,-1) node[below]{$b$};
\draw[-latex] (9,1.25) -- node[above]{$V_1$} (0,1.25);
\draw[-latex] (9,1.25) -- node[above]{$V_2$} (18,1.25);
\draw[latex-latex] (9,-.25) --node[below]{$\LCP_p(U_1,V_1)$} (5,-.25);
\draw[latex-latex] (9,-.25) --node[below]{$\LCP_q(U_2,V_2)$} (14,-.25);
\end{scope}
\end{tikzpicture}
\vspace{-.5cm}
\end{center}
\caption{If $\LCF_k(X,Y)\ge \ell$, then there exist $(U_1,U_2)\in \Pairs_\ell(X)$ and $(V_1,V_2)\in \Pairs_\ell(Y)$
such that $\LCF_k(X,Y)=\LCP_{p}(U_1,U_2)+\LCP_{q}(V_1,V_2)$ for some $p+q=k$.}\label{fig:red}
\end{figure}

\begin{lemma}\label{lem:red}
If $\LCF_k(X,Y)\ge \ell$, then
\begin{align*}
&\LCF_k(X,Y)=\\ &\quad\max_{p+q=k}\;\{\LCP_p(U_1,V_1)+\LCP_q(U_2,V_2)\;:\;
(U_1,U_2)\in \Pairs_\ell(X),\, (V_1,V_2)\in \Pairs_\ell(Y)\}.
\end{align*}
\end{lemma}
\begin{proof}
First, assume that $(U_1,U_2)\in \Pairs_\ell(X)$, $(V_1,V_2)\in \Pairs_\ell(Y)$, and $k=p+q$.
Let $\tilde{U}_1$ and $\tilde{V}_1$ be prefixes of $U_1$ and $V_1$ (respectively)
of length $\LCP_p(U_1,V_1)$, and let $\tilde{U}_2$ and $\tilde{V}_2$
be prefixes of $U_2$ and $V_2$ (respectively) of length $\LCP_q(U_2,V_2)$.
Observe that $\tilde{U}_1^R \tilde{U}_2$ is a factor of $X=U_1^RU_2$
$\tilde{V}_1^R \tilde{V}_2$ is a factor of $Y=V_1^R V_2$.
Moreover,
$$\HD(\tilde{U}_1^R \tilde{U}_2, \tilde{V}_1^R \tilde{V}_2)= \HD(\tilde{U}_1,\tilde{U}_2)+\HD(\tilde{V}_1,\tilde{V}_2)\le p+q= k.$$
Consequently,
$$\LCF_k(X,Y)\le |\tilde{U}_1^R \tilde{U}_2|=|\tilde{U}_1| + |\tilde{U}_2|=\LCP_p(U_1,V_1)+\LCP_q(U_2,V_2).$$
This concludes the proof of the claimed upper bound on $\LCF_k$.

Next, let $X[i\dd i']$ and $Y[j\dd j']$
be an optimal pair of factors; see \cref{fig:red}. They satisfy $$|X[i\dd i']|=|Y[j\dd j']| = \LCF_k(X,Y)\ge \ell\quad\text{and}\quad d_H(X[i\dd i'], Y[j\dd j'])\le k.$$
Denote $a=i+h(i,j)$ and $b=j+h(i,j)$,
where $h$ is the shift function associated with the $l$-cover $\S(l)$.
Note that $a \in [i\dd i']\cap \S(l)$ and $b\in [j\dd j']\cap \S(l)$.
Consequently, $(U_1,U_2):= ((X[\dd a-1])^R,X[a\dd])\in \Pairs_\ell(X)$ and $(V_1,V_2):= ((Y[\dd b-1])^R,Y[b\dd])\in \Pairs_\ell(Y)$.
Moreover, $$k\ge \HD(X[i\dd i'],Y[j\dd j'])=\HD(X[i\dd a-1],Y[j\dd b-1])+\HD(X[a\dd i'],Y[b\dd j']).$$
Therefore, for  $p = \HD(X[i\dd a-1],Y[j\dd b-1])$ and $q=k-p$, we have
$$\LCP_p(U_1,V_1)+\LCP_{q}(V_2,U_2)\ge |X[i\dd a-1]|+|X[a\dd i']| = |X[i\dd i']|= \LCF_k(X,Y).$$
This concludes the proof.
 \end{proof}

%
%

\section{The case of $k=0$ and of $k= 1$ and $\sigma=2$}
In this section, as a warmup, we show how the \problem can be used to solve two special cases of \kLCF$(X,Y,\ell)$. 
Then in \cref{sec:main} we explain how it can be used to solve the problem in full generality.

In order to solve \kLCF($X,Y,\ell$) for $k=0$,
we observe that, by Lemma~\ref{lem:red}, if
$\LCF_0(X,Y)\ge \ell$,
 then $\LCF_0(X,Y)= \maxPairLCP(\Pairs_\ell(X),\Pairs_\ell(Y))$.
 Thus, we simply build the joint suffix tree $\T$ of $X$, $Y$, $X^R$, and $Y^R$,
and we solve the appropriate instance of \problem.

The preprocessing time is clearly $\Oh(n)$,
while solving the \problem takes $\Oh(n+n\log n/\sqrt{\ell})$ time, which is $\Oh(n)$
provided that $\ell=\Omega(\log^2 n)$.

 For $k\ge 1$, we would ideally like to extend the family $\Pairs_\ell(S)$ to $\Pairs_\ell^{(k)}(S)$
replacing the suffixes and reversed prefixes of $S$ with their approximate copies
so that
$$\LCF_k(X,Y)=\maxPairLCP(\Pairs_{\ell}^{(k)}(X),\Pairs_{\ell}^{(k)}(Y)).$$

A very naive solution would be to extend the alphabet $\Sigma$ to $\Sh$ adding a symbol $ \$\notin \Sigma$,
and for each $(S_1,S_2)\in \Pairs_\ell(S)$ to replace an arbitrary subset of $k$ symbols with $ \$$'s.
However, this results in $\binom{n}{k}$ copies of each $(S_1,S_2)\in \Pairs_\ell(S)$, which is by far too much.

Our approach is therefore based on the technique of Cole et al.~\cite{DBLP:conf/stoc/ColeGL04},
which has already been used in the context of the \kLCFFull problem by Thankachan et al.~\cite{DBLP:journals/jcb/ThankachanAA16}.
It allows us to reduce the number of approximate copies of each $(S_1,S_2)\in \Pairs_\ell(S)$ to $\Oh(\log^k n)$.
However, the sets $\Pairs_\ell^{(k)}(X)$ and $\Pairs_\ell^{(k)}(Y)$ cannot be constructed independently,
and we actually have to build several pairs of such sets rather just one.

Below, we explain the main points for $k=1$ and $\sigma=2$. The description is illustrated in Example~\ref{ex:big} in the appendix.

Let $\F$ be a family consisting of the suffixes of $X$, $X^R$, $Y$, and $Y^R$, appearing $\Pairs_\ell(X)$ or $\Pairs_\ell(Y)$.
We apply the heavy-light decomposition on the compacted trie $\T(F)$; this technique can be summarized as follows:
\begin{fact}[Tarjan~\cite{DBLP:journals/jacm/Tarjan79}]\label{fct:hld}
If $T$ is a rooted tree, then in linear time we can mark some edges in $T$ as \emph{light} so that:
\begin{itemize}
  \item each node has at most one outgoing edge which is \emph{not} light,
  \item each root-to-leaf path contains $\Oh(\log |T|)$ light edges.
\end{itemize}
\end{fact}

Next, for each string $F\in \F$, we construct a set $N(F)$ consisting of $F$ and any
string which can be obtained from $F$ by flipping the first symbol of a single light edge
on the path representing $F$ in $\T(\F)$. By \cref{fct:hld}, we have $|N(F)|=\Oh(\log |\F|)=\Oh(\log n)$.

Let us denote $N_0(F)=\{F\}$ and $N_1(F)=N(F)$.
These sets have been constructed so that they enjoy the following crucial property:
\begin{lemma}\label{lem:lcpd1}
If $F,G\in \F$, then
$$\LCP_1(F,G)=\max_{d_1+d_2=1} \{\LCP(F',G') : F'\in N_{d_1}(F),  G'\in N_{d_2}(G)\}.$$
\end{lemma}
\begin{proof}
First, let us bound $\LCP_1(F,G)$ from above.
Let $p=\LCP(F',G')$ be the maximum on the right-hand side,
We have
\begin{multline*}
\HD(F[\dd p], G[\dd p])=\HD(F[\dd p], F'[\dd p])+\HD(G'[\dd p], G[\dd p])\le\\\le \HD(F,F')+\HD(G',G)\le d_1+d_2=1.
\end{multline*}
Consequently, $\LCP_1(F,G)\ge p$ as claimed.

To bound $\LCP_1(F,G)$ from above, let us consider
terminal nodes $v_F$ and $v_G$ in $\T(\F)$ representing $F$ and $G$, respectively,
and their lowest common ancestor $v$.
If $v=v_F$ or $v=v_G$, then $\LCP_1(F,G)=\LCP(F,G)$ and the claimed bound holds due to $F\in N_0(F)$ and $G\in N_1(G)$ (and vice versa).
Otherwise, the edge from $v$ towards $v_F$ or the edge from $v$ towards $v_G$ has to be light (according to \cref{fct:hld}).
If the former edge is light, then $N_1(F)$ contains a string $F'$ obtained from $F$
by flipping the first character on that edge. Such a string $F'$ satisfies
$\LCP_1(F,G)=\LCP(F',G)$, so the claimed bound holds due to $G\in N_0(G)$.
Symmetrically, if the edge towards $v_G$ is light, then
$\LCP_1(F,G)=\LCP(F,G')$ for some $G'\in N_1(G)$. 
\end{proof}

For $S\in \{X,Y\}$ and $d\in \{0,1\}$, let us define
$$\Pairs^{(d)}_\ell(S) = \bigcup_{\substack{(U_1,U_2)\in \Pairs_\ell(X)\\d_1+d_2 = 1}} \{(U'_1,U'_2) : U'_1 \in N_{d_1}(U_1), U'_2\in N_{d_2}(U_2)\}.$$
Observe that $\Pairs^{(0)}_\ell(S)=\Pairs_\ell(S)$, whereas the set $\Pairs^{(1)}_\ell(S)$ satisfies $|\Pairs^{(1)}_\ell(S)| = \Oh(|\Pairs_\ell(S)|\log |\F|)=\Oh(n\log n / \sqrt{\ell})$.
\Cref{lem:red,lem:lcpd1} yield the following
\begin{corollary}\label{cor:lcpd1}
If $\LCF_1(X,Y)\ge \ell$, then
$$\LCF_1(X,Y) = \max_{k_1+k_2=1}\maxPairLCP(\Pairs^{(k_1)}_\ell(X),\Pairs^{(k_2)}_\ell(Y)).
$$
\end{corollary}
\begin{proof}
By \cref{lem:red}, we have $\LCF_1(X,Y) = \LCP_p(U_1,V_1)+\LCP_q(U_2,V_2)$ for some $(U_1,U_2)\in \Pairs_{\ell}(X)$,
$(V_1,V_2)\in \Pairs_\ell(Y)$, and $p+q=1$. 
\cref{lem:lcpd1} yields that $\LCP_p(U_1,V_1)=\LCP(U'_1,V'_1)$ for some $U'_1\in N_{p_1}(U_1)$ and $V'_1\in N_{p_2}(V_1)$
such that $p=p_1+p_2$. Similarly, $\LCP_q(U_2,V_2)=\LCP(U'_2,V'_2)$ for some $U'_2\in N_{q_1}(U_2)$ and $V'_2\in N_{q_2}(V_2)$.
Observe that $(U'_1,U'_2)\in \Pairs^{(p_1+q_1)}_{\ell}(X)$ and $(V'_1,V'_2)\in \Pairs^{(p_2+q_2)}_{\ell}(Y)$,
so $$\LCF_1(X,Y)\le \maxPairLCP(\Pairs^{(k_1)}_\ell(X),\Pairs^{(k_2)}_\ell(Y))$$ for $k_i=p_i+q_1$
(which satisfy $k_1+k_2 = p+q=1$, as claimed).

Next, suppose that $(U'_1,U'_2)\in \Pairs^{(k_1)}_\ell(X)$ and $(V'_1,V'_2)\in \Pairs^{(k_2)}_\ell(Y)$.
We shall prove that $\LCF_1(X,Y)\ge \LCP(U'_1,V'_1)+\LCP(U'_2,V'_2)$.
Note that $U'_1\in N_{p_1}$ and $U'_2\in N_{q_1}(U_2)$ for some $p_1+q_1=k_1$ and $(U_1,U_2)\in \Pairs_\ell(X)$;
symmetrically, $V'_1\in N_{p_2}$ and $V'_2\in N_{q_2}(V_2)$ for some $p_2+q_2=k_2$ and $(V_1,V_2)\in \Pairs_\ell(Y)$.
By \cref{lem:lcpd1}, $\LCP(U'_1,V'_1)\le \LCP_{p_1+p_2}(U_1,V_1)$ and $\LCP(U'_2,V'_2)\le \LCP_{q_1+q_2}(U_2,V_2)$.
Hence, the claimed bound holds due to \cref{lem:red}:
$$\LCF_1(X,Y)\ge \LCP_{p_1+p_2}(U_1,V_1)+\LCP_{q_1+q_2}(U_2,V_2)\ge \LCP(U'_1,V'_1)+\LCP(U'_2,V'_2).$$
This concludes the proof.
\end{proof}

Consequently, it suffices to solve two instances of \problem,
with $(\P,\Q)$ equal to  $(\Pairs^{(0)}_\ell(X),\Pairs^{(1)}_\ell(Y))$ and $(\Pairs^{(1)}_\ell(X),\Pairs^{(0)}_\ell(Y))$, respectively.

\begin{prop}
The problem \kLCF$(X,Y,\ell)$ for $k=1$ and binary alphabet can be solved in $\Oh(n + n\log^2 n /\sqrt{\ell})$
time. If $\ell=\Omega(\log^4 n)$, this running time is $\Oh(n)$.
\end{prop}
\begin{proof}
First, we build the sets $\Pairs_\ell(X)$ and $\Pairs_{\ell}(Y)$.
Next, we construct the joint suffix tree of strings $X$, $Y$, $X'$, $Y'$ (along with a component for constant-time LCA queries)
and we extract the compacted trie $\T(\F)$ of the family $\F$.
Then, we process light edges on $\T(\F)$ (determined by \cref{fct:hld}).
For each light edge $e$, we traverse the corresponding subtree and for each  terminal node (representing $F\in \F$),
we insert to $N(F)$ a string $F'$ obtained from $F$ by flipping the first character represented by $e$.
Technically, in $N(F)$ we just store the set of positions for which $F$ should be flipped to obtain $F'$.

To compute the compacted trie $\T(\F')$ of a family $\F'=\bigcup_{F\in \F}N(F)$,
we sort the strings in $F'\in \F'$ using a comparison-based algorithm.
Next, we extend the representation of $N(F)$ so that each $F'\in N(F)$ stores a pointer to the corresponding terminal node in $\T(\F')$.
This way, we can generate sets $\Pairs^{d}_\ell(S)$ for $d\in \{0,1\}$ and $S\in \{X,Y\}$
with strings represented as pointers to terminal nodes of $\T(\F')$.
Finally, we solve two instances of \problem according to \cref{cor:lcpd1}.

We conclude with the running-time analyis.
In the preprocessing, we spend $\Oh(n)$ time construct the joint suffix tree.
Then, applying \cref{fct:hld} to build the sets $N(F)$ for $F\in \F$ takes $\Oh(|\F|\log |\F|)=\Oh(n\log n / \sqrt{\ell})$
time. We spend further  $\Oh(|\F'|\log |\F'|)=\Oh(n\log^2n / \sqrt{\ell})$ time to construct $\T(\F')$.
Since $|\Pairs^{(d)}_\ell(S)|=\Oh(|\F|\log |\F|)$ for $d\in \{0,1\}$ and $S\in \{X,Y\}$,
the time to solve both instances of the \problem is also $\Oh(n\log^2n / \sqrt{\ell})$ (see Lemma~\ref{lem:problem}).
Hence, the overall time complexity is  $\Oh(n+n\log^2n / \sqrt{\ell})$.
\end{proof}

\section{Arbitary $k$ and $\sigma$}\label{sec:arbitrary}
In this section, we describe the core concepts of our solution for arbitrary number of mismatches $k$
and alphabet size $\sigma$. They depend heavily on the ideas behind the $\Oh(n \log^k n)$-time solution to \kLCF\  \cite{DBLP:journals/jcb/ThankachanAA16}, which originate in approximate indexing \cite{DBLP:conf/stoc/ColeGL04}.

\begin{definition}
Consider strings $U,V\in \Sigma^*$ and an integer $d\ge 0$.
We say that strings $U',V'\in \Sh^*$ form a $(U,V)_d$-pair if
\begin{itemize}
  \item $|U'|=|U|$ and $|V'|=|V|$;
  \item if $i>\LCP_d(U,V)$ or $U[i]=V[i]$, then $U'[i]=U[i]$ and $V'[i]=V[i]$;
  \item otherwise, $U'[i]=V'[i]\in \{U[i],V[i],\$\}$.
\end{itemize}
\end{definition}

\begin{definition}
Consider a finite family of strings $\F\sub \Sigma^*$.
We say that sets $N(F)\sub \Sh^*$ for $F\in \F$ form a \emph{$k$-complete family}
if for every $U,V\in \F$ and $0\le d \le k$, there exists a $(U,V)_d$-pair $U',V'$ with $U'\in N(U)$ and $V'\in N(V)$.
\end{definition}
\begin{remark}
A simple (yet inefficient) way to construct a $k$-complete family is to include
in $N(F)$ all strings which can be obtained from $F$ by replacing up to $k$ characters with $ \$$'s. An example of a more efficient family is shown in Table~\ref{tab:XY1} in the appendix.
\end{remark}

The following lemma states a property of $k$-complete families that we will use in the algorithm.
For $F\in \F$ and $0\le d \le k$, let us define $N_d(F)=\{F'\in N(F) : \HD(F,F')\le d\}$.
Moreover, for a half-integer $d'$, $0\le d'\le d$, let $$N_{d,d'}(F)=\{F'\in N_d(F) : \HD(F,F')-\tfrac12\#_{\$}(F')\le d'\}.$$
\begin{lemma}\label{lem:lcpd}
Let $N(F)$ for $F\in \F$ be a $k$-complete family.
If $F_1,F_2 \in \F$ and $0\le d \le k$, then
$$\LCP_d(F_1,F_2) = \max_{\substack{d_1+d_2= d \\F'_i \in N_{d,d_i}(F_i)}} \LCP(F'_1,F'_2)=\max_{\substack{\floor{d_1+d_2}\le d \\F'_i \in N_{k,d_i}(F_i)}} \LCP(F'_1,F'_2).$$
\end{lemma}
\begin{proof}
We shall prove that
$$\max_{\substack{d_1+d_2= d \\F'_i \in N_{d,d_i}(F_i)}} \LCP(F'_1,F'_2) \ge \LCP_d(F_1,F_2) \ge \max_{\substack{\floor{d_1+d_2}\le d \\F'_i \in N_{k,d_i}(F_i)}}(F'_1,F'_2).$$
This is sufficient due to the fact that $N_{d,d'}(F)$ is monotone with respect to both $d$ and $d'$.

For the first inequality, observe that (by definition of a $k$-complete family)
 the sets $N(F_1)$ and $N(F_2)$ contain an $(F_1,F_2)_{d}$-pair $(F'_1, F'_2)$.
Let $P$ be the longest common prefix of $F'_1$ and $F'_2$ ($|P| = \LCP_k(F_1,F_2)$)
and recall that by definition $F'_i = P F_i[|P|+1\dd]$.
Moreover, let $d_i = \HD(F_i,F'_i)-\frac12\#_{\$}(F'_i)$ so that  $F'_i\in N_{d,d_i}(F_i)$.
Consequently,
$$
d\ge \HD(F_1[\dd |P|],F_2[\dd |P|])= \HD(F_1,F'_1)+\HD(F_2,F'_2)- \#_{\$}(P) = d_1+d_2.
$$
If $d>d_1+d_2$, we may increase $d_1$ or $d_2$.

For the second inequality, suppose that $F'_i \in N_{k,d_i}(F_i)$ for $\floor{d_1+d_2} \le d$.
Let $P$ be the longest common prefix of $F'_1$ and $F'_2$.
Then
\begin{multline*}
\HD(F_1[\dd |P|],F_2[\dd |P|])\,\le\, \HD(F_1[\dd |P|],P)+\HD(F_2[\dd |P|],P])-\#_{\$}(P)\le \\
\le\, \HD(F_1,F_1')-\#_{\$}(F_1)+\HD(F_2,F_2')-\#_{\$}(F_2)+\#_{\$}(P) \le d_1 + d_2.
\end{multline*}
Consequently, $\HD(F_1[\dd |P|],F_2[\dd |P|]) = \floor{\HD(F_1[\dd |P|],F_2[\dd |P|])}\le \floor{d_1+d_2}\le d$,
as claimed.
\end{proof}

In the algorithms, we represent a $k$-complete family using the compacted trie $\T(\F')$ of the union $\F'=\bigcup_{F\in \F}N(F)$.
Its terminal nodes $F'$ are marked with a subset of strings $F\in \F$ for which $F'\in N(F)$;
for convenience we also store $\#_\$(F')$ and $\HD(F,F')$.
Each edge is labeled by a factor of $F\in \F$, perhaps prepended by $ \$$.

Our construction of a $k$-complete family is based on the results of \cite{DBLP:conf/stoc/ColeGL04,DBLP:journals/jcb/ThankachanAA16},
but we provide a self-contained proof in the appendix.
\begin{restatable}[see also~\cite{DBLP:conf/stoc/ColeGL04,DBLP:journals/jcb/ThankachanAA16}]{prop}{thmtrans}\label{prop:trans}
Let $\F\sub\Sigma^*$ be a finite family of strings
and let $k\ge 0$ be an integer.
There exists a $k$-complete family $N$
such that $|N_d(F)|\le 2^d \binom{\log |\F| + d}{d}$ for each $F\in \F$ and $0\le d \le k$.
Moreover, the compacted trie $\T(\F')$ can be constructed in $\Oh(2^k|\F|\binom{\log |\F|+k+1}{k+1})$ time
provided constant-time $\LCP$ queries for suffixes of the strings $F\in \F$.
\end{restatable}

\begin{remark}
  The 1-complete family from Table~\ref{tab:XY1} is a subset of the family constructed by the algorithm that is behind \cref{prop:trans}.
\end{remark}

\section{Main Result}\label{sec:main}

Let $\F$ be a family of suffixes and reverse prefixes of $X$ and $Y$ occurring in $\Pairs_\ell(X)$ or $\Pairs_{\ell}(Y)$,
and let us fix a $k$-complete family $N(F) : F\in \F$.
For a half-integer $k'$, $0\le k' \le k$, and a string $S\in \{X,Y\}$ let us define
$$\Pairs^{(k,k')}_{\ell}(S) = \bigcup_{(U_1,U_2)\in \Pairs_\ell(S)} \{(U'_1,U'_2) : U'_i \in N_{d_i,d'_i}(U_i), k=d_1+d_2, k'=d'_1+d'_2\}.$$
To bound the size of $\Pairs^{(k,k')}_{\ell}(S)$, we observe that for $d_1+d_2=k$ and $k=\Oh(\log n)$
$$|N_{d_1}(U_1)|\cdot |N_{d_2}(U_2)|\le 2^{k}\tbinom{\log |\F|+d_1}{d_1}\tbinom{\log |\F|+d_2}{d_2}=\tfrac{2^{\Oh(k)}\log^k|\F|}{k^k}.$$
Hence, $|\Pairs^{(k,k')}_{\ell}(S)|=\frac{2^{\Oh(k)}|\F|\log^k|\F|}{k^k \sqrt{\ell}}$.
Combining \cref{lem:red,lem:lcpd}, we obtain the following.
\begin{corollary}\label{cor:lcpd}
If $\LCF_k(X,Y)\ge \ell$,
then
$$\LCF_k(X,Y)=\max_{k_1+k_2= k}\maxPairLCP(\Pairs^{(k,k_1)}_{\ell}(X), \Pairs^{(k,k_2)}_{\ell}(Y)).$$
\end{corollary}
\begin{proof}
By \cref{lem:red}, there exist $(U_1,U_2)\in \Pairs_\ell(X)$, $(V_1,V_2)\in \Pairs_\ell(Y)$, and $p+q=k$
such that $\LCF_k(X,Y)=\LCP_p(U_1,V_1)+\LCP_q(U_2,V_2)$.
\Cref{lem:lcpd} further yields existence of half-integers $p'_1+p'_2 \le p$
and $q'_1+q'_2 \le q$ such that $\LCF_k(X,Y)=\LCP(U'_1,V'_1)+\LCP(U'_2,V'_2)$
for some
$U'_1 \in N_{p,p'_1}(U_1)$, $V'_1\in N_{p,p'_2}(V_1)$, $U'_2 \in N_{q,q'_1}(U_2)$, and $V'_2 \in N_{q,q'_2}(V_2)$.

We set $k'_1 = p'_1 + q'_1$ and $k'_2=k-k'_1 \ge p'_2+q'_2$
so that $(U'_1,U'_2)\in \Pairs_{\ell}^{(k,k'_1)}(X)$
and $(V'_1,V'_2)\in \Pairs_{\ell}^{(k,p'_2+q'_2)}(Y)\sub \Pairs_{\ell}^{(k,k'_2)}(Y)$.
Consequently,
$$\LCF_k(X,Y)\le \maxPairLCP( \Pairs_{\ell}^{(k,k'_1)}(X), \Pairs_{\ell}^{(k,k'_2)}(Y)),$$
which concludes the proof of the upper bound on $\LCF_k(X,Y)$.

For the lower bound, we shall prove
that $\LCF_k(X,Y) \ge \LCP(U'_1,V'_1)+\LCP(U'_2,V'_2)$
for all $(U'_1,U'_2)\in \Pairs^{(k,k'_1)}_{\ell}(X)$ and $(V'_1,V'_2)\in \Pairs^{(k,k'_2)}_{\ell}(Y)$
such that $k'_1+k'_2\le k$.
By definition of $\Pairs^{(k,k'_1)}_\ell$, there exist $(U_1,U_2)\in \Pairs_\ell(X)$
such that $U'_1\in N_{k,p'_1}(U_1)$ and $U'_2\in N_{k,q'_1}(U_2)$ for half-integers $p'_1 + q'_1 \le k'_1$.
Similarly, there exist  $(V_1,V_2)\in \Pairs_\ell(Y)$
such that $V'_1\in N_{k,p'_2}(V_1)$ and $V'_2\in N_{k,q'_2}(V_2)$ for half-integers $p'_2 + q'_2 \le k'_2$.
We set $p = \floor{p'_1+p'_2}$ and $q=\floor{q'_1+q'_2}$,
and observe that $\LCP_{p}(U_1,V_1)\ge \LCP(U'_1,V'_1)$ as well as $\LCP_{q}(U_2,V_2)\ge \LCP(U'_2,V'_2)$ due to \cref{lem:lcpd}.
Now, \cref{lem:red} yields
$\LCF_k(X,Y)\ge \LCP_{p}(U_1,V_1)+\LCP_{q}(U_2,V_2) \ge \LCP(U'_1,V'_1)+\LCP(U'_2,V'_2)$, as desired.
\end{proof}

\begin{theorem}\label{thm:main}
For $k= \Oh(\log n)$,
the \kLCF$(X,Y,\ell)$ problem can be solved in time $\Oh(n + \tfrac{2^{\Oh(k)}n\log^{k+1}{n}}{k^k\sqrt{\ell}})$.
For $k=\Oh(1)$, this running time becomes $\Oh(n+\tfrac{n\log^{k+1}n}{\sqrt{\ell}})$.
\end{theorem}
\begin{proof}
First, we build the joint suffix tree of $X$, $X^R$, $Y$, and $Y^R$,
as well as the family $\F$. A component for the LCA queries on the suffix tree
lets us compare any suffixes of $F\in \F$ in constant time.
This allows us to build the $k$-complete family $N(F) : F\in \F$,
represented as a compacted trie of $\F' :=\bigcup\{N(F) : F\in \F\}$ using \cref{prop:trans}.
Next, we construct the sets $\Pairs_{\ell}^{(k,k')}(X)\sub (\F')^2$ and $\Pairs_{\ell}^{(k,k')}(Y)\sub (\F')^2$
for $k'=0,\frac12, \ldots, k-\frac12, k$, and solve the $2k+1$ instances of \problem,
as specified in \cref{cor:lcpd}.

We conclude with running-time analysis.
Preprocessing takes $\Oh(n)$ time,
and the procedure of \cref{prop:trans} runs in  $\Oh(2^k|\F|\binom{\log |\F|+k+1}{k+1})=\frac{2^{\Oh(k)}n\log^{k+1}n}{k^k \sqrt{\ell}}$
time.
We have $\Pairs_{\ell}^{k,k'}(X) =\frac{2^{\Oh(k)}n\log^{k}n}{k^k \sqrt{\ell}}$,
so solving all instances of  also takes $\frac{2^{\Oh(k)}n\log^{k+1}n}{k^k \sqrt{\ell}}$ time (Lemma~\ref{lem:problem}).
The overall running time is therefore as claimed.
\end{proof}

In particular, for $k=\Oh(\log n)$, there exists $\ell_0 = \tfrac{2^{\Oh(k)}\log^{2k+2} n}{k^{2k}}$ such that
\kLCF$(X,Y,\ell)$ can be solved in $\Oh(n)$ time for $\ell \ge \ell_0$.
For $k=\Oh(1)$, we have $\ell_0 = \Oh(\log^{2k+2} n)$,
while for $k=o(\log n)$, we have $\ell_0 = n^{o(1)}$.
We arrive at the main result.

\begin{corollary}
  \textsc{Long }\kLCF\ can be solved in $\Oh(n)$ time.
\end{corollary}

\bibliographystyle{plainurl}
\bibliography{references}

\newpage
\appendix
\section{Examples}
This section contains additional examples related to our application of the technique of Cole et al.~\cite{DBLP:conf/stoc/ColeGL04}.

\begin{example}\label{ex:big}
Let us consider the \oLCF$(X,Y,\ell)$ problem for $X= \texttt{bbaaabb}$, $Y=\texttt{abababa}$, $\ell=5$.

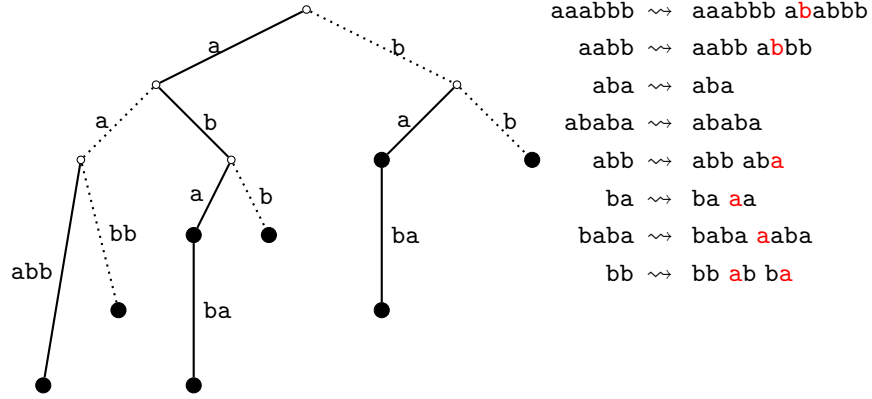
\begin{center}
\begin{tikzpicture}[yscale=-1]
\begin{scope}
\begin{scope}[every node/.style={circle,minimum size=.1cm,draw=black,inner sep = 0pt}, final/.style={minimum size=.2cm, fill=black}]
\draw(0,0) node[circle](eps){};
\draw(-2,1) node(a){};
\draw(-3,2) node(aa){};
\draw(-3.5,5) node[final](aaabb){};
\draw(-2.5,4) node[final](aabb){};
\draw(-1,2) node(ab){};
\draw (-1.5,3) node[final](aba){};
\draw (-1.5,5) node[final](ababa){};
\draw (-0.5,3) node[final](abb){};
\draw(2,1) node(b){};
\draw(1,2) node[final](ba){};
\draw(1,4) node[final](baba){};
\draw(3,2) node[final](bb){};
\end{scope}

\draw[thick] (eps) --node[left]{$\texttt{a}$} (a);
\draw[thick,dotted] (a) --node[left]{$\texttt{a}$} (aa);
\draw[thick] (aa) --node[left]{$\texttt{abb}$} (aaabb);
\draw[thick,dotted] (aa) --node[right]{$\texttt{bb}$} (aabb);
\draw[thick] (a) --node[right]{$\texttt{b}$} (ab);
\draw[thick] (ab) --node[left]{$\texttt{a}$} (aba);
\draw[thick] (aba) --node[right]{$\texttt{ba}$} (ababa);
\draw[thick,dotted] (ab) --node[right]{$\texttt{b}$} (abb);
\draw[thick,dotted] (eps) --node[right]{$\texttt{b}$} (b);
\draw[thick] (b) --node[left]{$\texttt{a}$} (ba);
\draw[thick] (ba) --node[right]{$\texttt{ba}$} (baba);
\draw[thick,dotted] (b) --node[right]{$\texttt{b}$} (bb);
\end{scope}

\begin{scope}[xshift=5cm]
\draw (0,0) node[left]{$\texttt{aaabbb}\:\leadsto$};
\draw (0,0) node[right]{$\texttt{aaabbb}\ \texttt{a\textcolor{red}{b}abbb}$};

\draw (0,1/2) node[left]{$\texttt{aabb}\:\leadsto$};
\draw (0,1/2) node[right]{$\texttt{aabb}\ \texttt{a\textcolor{red}{b}bb}$};

\draw (0,2/2) node[left]{$\texttt{aba}\:\leadsto$};
\draw (0,2/2) node[right]{$\texttt{aba}$};

\draw (0,3/2) node[left]{$\texttt{ababa}\:\leadsto$};
\draw (0,3/2) node[right]{$\texttt{ababa}$};

\draw (0,4/2) node[left]{$\texttt{abb}\:\leadsto$};
\draw (0,4/2) node[right]{$\texttt{abb}\ \texttt{ab\textcolor{red}{a}}$};

\draw (0,5/2) node[left]{$\texttt{ba}\:\leadsto$};
\draw (0,5/2) node[right]{$\texttt{ba}\ \texttt{\textcolor{red}{a}a}$};

\draw (0,6/2) node[left]{$\texttt{baba}\:\leadsto$};
\draw (0,6/2) node[right]{$\texttt{baba}\ \texttt{\textcolor{red}{a}aba}$};

\draw (0,7/2) node[left]{$\texttt{bb}\:\leadsto$};
\draw (0,7/2) node[right]{$\texttt{bb}\ \texttt{\textcolor{red}{a}b}\ \texttt{b\textcolor{red}{a}}$};

\end{scope}

\end{tikzpicture}
\end{center}
Let $\S(5)=\{x \in \mathbb{Z} : x\bmod 5 \in \{0,3,4\}\}$.
We have $$\Pairs_5(X) = \{(\texttt{bb},\texttt{aaabb}),(\texttt{abb},\texttt{aabb}),(\texttt{aabb},\texttt{abb})\},$$
and $$\Pairs_5(Y) = \{(\texttt{ba},\texttt{ababa}),(\texttt{aba},\texttt{baba}),(\texttt{baba},\texttt{aba})\}.$$
The compacted trie $\T(\F)$ is illustrated above with light edges dotted.
As a result,
\begin{multline*}
\Pairs_5^{(1)}(X)=\Pairs_5(X)\;\cup\\
\{(\texttt{\textcolor{red}{a}b},\texttt{aaabb}),(\texttt{b\textcolor{red}{a}},\texttt{aaabb}),(\texttt{bb},\texttt{a\textcolor{red}{b}abb}),
(\texttt{ab\textcolor{red}{a}},\texttt{aabb}),(\texttt{abb},\texttt{aa\textcolor{red}{a}b}),(\texttt{aabb},\texttt{ab\textcolor{red}{a}}),(\texttt{aa\textcolor{red}{a}b},\texttt{abb})\}
\end{multline*}
and
$$
\Pairs_5^{(1)}(Y)=\Pairs_5(Y)\cup \{(\texttt{\textcolor{red}{a}a},\texttt{ababa}),(\texttt{aba},\texttt{\textcolor{red}{a}aba}),(\texttt{\textcolor{red}{a}aba},\texttt{aba})\}
$$
Consequently,
$$\maxPairLCP(\Pairs_5^{(1)}(X),\Pairs_5^{(0)}(Y))=\LCP(\texttt{bb},\texttt{ba})+\LCP(\texttt{a\textcolor{red}{b}abb},\texttt{ababa})=5,$$
and 
\begin{multline*}\maxPairLCP(\Pairs_5^{(0)}(X),\Pairs_5^{(1)}(Y))=\LCP(\texttt{abb},\texttt{aba})+\LCP(\texttt{aabb},\texttt{\textcolor{red}{a}aba})=\\=\LCP(\texttt{aabb},\texttt{\textcolor{red}{a}aba})+\LCP(\texttt{abb},\texttt{aba})=5.\end{multline*}
All these pairs correspond to factors $\texttt{baaab}$ of $X$ and $\texttt{babab}$ of $Y$.
\end{example}

\begin{table}[h!]
\centering
\begin{tabular}{|l|l|l|l|l|l|}
\hline
&\texttt{b}&\texttt{cb}&\texttt{acb}&\texttt{bacb}&\texttt{abacb}\\
\hline
\makecell[l]{\\\texttt{abacb}}&\makecell[l]{\texttt{\underline{\textcolor{red}{a}}}\\\texttt{\underline{a}bacb}}&\makecell[l]{\texttt{\underline{\textcolor{red}{a}b}}\\\texttt{\underline{ab}acb}}&\makecell[l]{\texttt{\underline{a\textcolor{red}{b}}b}\\\texttt{\underline{ab}acb}}&\makecell[l]{\texttt{\underline{\textcolor{red}{a}}acb}\\\texttt{\underline{a}bacb}}&\makecell[l]{\texttt{\underline{abacb}}\\\texttt{\underline{abacb}}}\\
\hline
\makecell[l]{\\\texttt{bacb}}&\makecell[l]{\texttt{\underline{b}}\\\texttt{\underline{b}acb}}&\makecell[l]{\texttt{\underline{\textcolor{red}{\$}}b}\\\texttt{\underline{\textcolor{red}{\$}}acb}}&\makecell[l]{\texttt{\underline{a}cb}\\\texttt{\underline{\textcolor{red}{a}}acb}}&\makecell[l]{\texttt{\underline{bacb}}\\\texttt{\underline{bacb}}}&\makecell[l]{\texttt{\underline{a}bacb}\\\texttt{\underline{\textcolor{red}{a}}acb}}\\\hline
\makecell[l]{\\\texttt{acb}}&\makecell[l]{\texttt{\underline{\textcolor{red}{a}}}\\\texttt{\underline{a}cb}}&\makecell[l]{\texttt{\underline{\textcolor{red}{a}}b}\\\texttt{\underline{a}cb}}&\makecell[l]{\texttt{\underline{acb}}\\\texttt{\underline{acb}}}&\makecell[l]{\texttt{\underline{\textcolor{red}{a}}acb}\\\texttt{\underline{a}cb}}&\makecell[l]{\texttt{\underline{ab}acb}\\\texttt{\underline{a\textcolor{red}{b}}b}}\\
\hline
\makecell[l]{\\\texttt{cb}}&\makecell[l]{\texttt{\underline{\textcolor{red}{\$}}}\\\texttt{\underline{\textcolor{red}{\$}}b}}&\makecell[l]{\texttt{\underline{cb}}\\\texttt{\underline{cb}}}&\makecell[l]{\texttt{\underline{a}cb}\\\texttt{\underline{\textcolor{red}{a}}b}}&\makecell[l]{\texttt{\underline{\textcolor{red}{\$}}acb}\\\texttt{\underline{\textcolor{red}{\$}}b}}&\makecell[l]{\texttt{\underline{ab}acb}\\\texttt{\underline{\textcolor{red}{a}b}}}\\
\hline
\makecell[l]{\\\texttt{b}}&\makecell[l]{\texttt{\underline{b}}\\\texttt{\underline{b}}}&\makecell[l]{\texttt{\underline{\textcolor{red}{\$}}b}\\\texttt{\underline{\textcolor{red}{\$}}}}&\makecell[l]{\texttt{\underline{a}cb}\\\texttt{\underline{\textcolor{red}{a}}}}&\makecell[l]{\texttt{\underline{b}acb}\\\texttt{\underline{b}}}&\makecell[l]{\texttt{\underline{a}bacb}\\\texttt{\underline{\textcolor{red}{a}}}}\\
\hline
\end{tabular}
\caption{A sample $1$-complete family for $\F=\{\texttt{abacb},\texttt{bacb},\texttt{acb},\texttt{cb},\texttt{b}\}$ (the suffixes of \texttt{abacb})
is $N(\texttt{b})=\{\texttt{a},\texttt{b},\texttt{\$}\}$, $N(\texttt{cb})=\{\texttt{ab},\texttt{cb},\texttt{\$b}\}$,
$N(\texttt{acb})=\{\texttt{abb},\texttt{acb}\}$, $N(\texttt{bacb})=\{\texttt{aacb},\texttt{bacb},\texttt{\$acb}\}$, and $N(\texttt{abacb})=\{\texttt{abacb}\}$.
The $(U,V)_1$-pairs for all $U,V\in \F$ are illustrated in the table above.
Observe that $\LCP_1(U,V)=\LCP(U',V')$ for the corresponding $(U,V)_1$-pair $(U',V')$.
Also, note that $\LCP_1(\texttt{acb},\texttt{cb})=1$ even though $\texttt{abb}\in N(\texttt{acb})$, $\texttt{ab}\in N(\texttt{cb})$,
and $\LCP(\texttt{abb},\texttt{ab})=2$.
 }\label{tab:XY1}
\end{table}

\newpage
\section{Proof of \cref{prop:trans}}

In this section we show an efficient construction of a $k$-complete family.


 We apply a recursive procedure that builds the subtree rooted at the node representing $P$.
 The input $\F_P$ consists of tuples $(S,F,b)$ such that
 $F\in \F$, $S$ is a suffix of $F$ of length $|S|=|F|-|P|$, and $b=k-\HD(F,PS)\ge 0$.
 Intuitively, the parameter $b$ can be seen as a ``budget'' of remaining symbol changes in the string that prevents exceeding the number $k$ of mismatches.
 In the first call we have $P=\eps$ and $\F_P = \{(F,F,k):F \in \F\}$.

 In the pseudocode below we state this procedure in an abstract way; afterwards we explain how to implement it efficiently.

\medskip
 \SetKwFunction{FR}{Generate}
\begin{algorithm}[H]
\SetKwProg{Fn}{Function}{ is}{end}
 \Fn{\FR{$P,\F_P$}}{
$h:=\text{a most frequent element of } \{S[1] : (S,F,b)\in \F_P \text{ and }S\ne \eps$\}\;
\ForEach(\tcp*[f]{$b=k-\HD(F,PS)\ge 0$}){$(S,F,b)\in \F_P$}{
\lIf{$S=\eps$}{$N(F):= N(F)\cup \{P\}$}
\Else{
$c := S[1]$\;
$\F_{Pc} := \F_{Pc}\cup \{\,(S[2\dd ],F,b)\,\}$\;
\If{$c \ne h$ \KwSty{and} $b>0$}{
$\F_{Ph} := \F_{Ph}\cup \{\,(S[2\dd],F,b-1)\,\}$\;$\F_{P\$} := \F_{P\$}\cup \{\,(S[2\dd],F,b-1)\,\}$\;
}
}
}
\ForEach{$c\in \Sigma\cup \{\$\}$ \KwSty{such that} $\F_{Pc}\ne \emptyset$}{
\FR{$Pc,\F_{Pc}$}\;
}
}
\label{alg:gen}
\caption{A recursive procedure inserting strings with prefix $P$ to sets $N(F)$.}
\end{algorithm}

\medskip
The proof of \cref{prop:trans} is divided into three claims that characterize the output of the above procedure.

\begin{claim}
For every $S,T\in \F$ and $0\le d \le k$, there exists an $(S,T)_d$-pair $(S',T')$ with $S'\in N(S)$ and $T'\in N(T)$.
\end{claim}
\begin{proof}
We first observe that the algorithm satisfies the following property:
\begin{observation}
If $(S,F,b)\in \F_P$, then $PS$ is eventually added to $N(F)$.
\end{observation}

Next, we inductively prove that if $(S,F,b),(T,F',b')\in \F_P$ and $b,b'\ge d$,
then there exists an $(S,T)_d$-pair $(S',T')$ such that $PS'$ is added to $N(F)$ and $PT'$ is added to $N(F')$.

We proceed by induction on $|S|+|T|$.
If $|S|=0$, $|T|=0$, or $d=0$, then $S=S'$ and $T=T'$. Moreover,
$PS$ is added to $N(F)$ and $PT$ is added to $N(F')$ by the previous claim.

Thus, below we assume that these three quantities are all positive.
If $S[1]=T[1]=c$, then $(S[2\dd],F,b)$ and $(T[2\dd],F',b')$
are added to $\F_{Pc}$ and, by the inductive hypothesis,
we have an $(S[2\dd],T[2\dd])_d$-pair $(S'',T'')$ with $PcS''\in N(F)$
and $PcT''\in N(F')$.
We observe that $cS'',cT''$ is an $(S,T)_d$-pair.
If $S[1] \ne T[1]$, $(S[2\dd],F,b-1)$ and $(T[2\dd],F',b'-1)$
are both added to $\F_{P\$}$ (if $S[1]\ne h$ and $T[1]\ne h$) or to $\F_{Ph}$ (otherwise).
In either case, by the inductive hypothesis we have an $(S[2\dd],T[2\dd])_{d-1}$ pair $(S'',T'')$ with $PcS''\in N(F)$
and $PcT''\in N(F')$. It suffices to observe that $(cS'',cT'')$ is then an $(S,T)_d$-pair.

Finally, we derive the lemma because $S,T\in \F$ implies that $(S,S,k),(T,T,k)\in \F_{\eps}$.
\end{proof}

\begin{claim}
For each $F\in \F$, we have $|N_d(F)|\le 2^d\binom{\log |\F| + d}{d}$.
\end{claim}
\begin{proof}
For each $P\in \Sh^*$ us define $N_{d,P}(F) = \{F' \in N_d(F) : P\text{ is a prefix of }F'\}$.
We inductively prove the following bound for decreasing $|P|$:
$$|N_{d,P}(F)| \le \begin{cases}
2^{b}\binom{\log |\F_P| + b}{b} & \text{if }(S,F,b+k-d)\in \F_P\text{ and }b\ge 0,\\
0 & \text{otherwise.}
\end{cases}$$
If $(S,F,b+k-d)\notin \F_P$ for $b\ge 0$, then $N_d(F)$ does not contain any string $F'$ with prefix $P$.
Thus, we focus on the case when $(S,F,b+k-d)\in \F_P$ for $b\ge 0$.

If $|P|=|F|$, then
$$|N_{d,P}(F)| = |\{P\}|=1 = 2^0 \tbinom{\log|\F|}{0}\le 2^b \tbinom{\log|\F|+b}{b},$$
so the claimed inequality holds.

Otherwise, let $h$ be defined as in \FR($P,\F_P$).
If $S[1]=h$, then we just insert $(S[2\dd],F,b+k-d)$ to $\F_{Ph}$.
Consequently,
$$|N_{d,P}(F)|= |N_{d,Ph}(F)| \le 2^{b}\tbinom{\log |\F_{Ph}|+b}{b}
\le  2^{b}\tbinom{\log |\F_{P}|+b}{b},$$
as claimed.

Otherwise, $(S[2\dd],F,b+k-d)$ is inserted to $\F_{PS[1]}$
and  $(S[2\dd],F,b+k-d-1)$ is inserted to $\F_{Pa}$ for $a\in \{h,\$\}$
provided that $b > 0$.
Moreover, we observe that $|\F_{PS[1]}|\le \frac12 |\F_{P}|$ due to $S[1]\ne h$.
Consequently,
\begin{align*}
|N_{d,P}(F)|&= |N_{d,PS[1]}(F)|+|N_{d,Ph}(F)|+|N_{P\$}(F)|\le\\
&\le 2^{b}\tbinom{\log |\F_{PS[1]}| + b}{b}
+  2^{b-1}\tbinom{\log |\F_{Ph}| +b-1}{b-1}
+ 2^{b-1}\tbinom{\log |\F_{P\$}| + b-1}{b-1}\le \\&\le
2^{b}\tbinom{\log (\frac12|\F_{P}|) + b}{b}
+2^{b-1}\tbinom{\log |\F_{P}| + b-1}{b-1}
+2^{b-1}\tbinom{\log |\F_{P}| + b-1}{b-1}=\\&=
2^{b}\left(\tbinom{\log |\F_{P}| - 1 + b}{b}+\tbinom{\log |\F_{P}| +b-1}{b-1}\right)
= 2^{b}\tbinom{\log |\F_P| + b}{b},
\end{align*}
as claimed.

Finally, we deduce for $P=\eps$ that:
$$|N_d(F)|=|N_{d,\eps}(F)|\le 2^{d}\tbinom{\log |\F_{\eps}|+d}{d}
=  2^{k}\tbinom{\log |\F|+d}{d}$$
due to $(F,F,d+k-d)\in \F_{\eps}$.
This concludes the proof.
\end{proof}

In the implementation of the procedure we use finger search trees~\cite{DBLP:conf/stoc/GuibasMPR77},
which maintain subsets of a linearly-ordered universe supporting constant-time queries.
Among many applications (see~\cite{DBLP:reference/crc/Brodal04} for a survey), 
they support the following two operations~\cite{DBLP:journals/iandc/HoffmanMRT86,DBLP:reference/crc/Brodal04}:
\begin{itemize}
  \item insert an element into a set $A$, which takes $\Oh(\log |A|)$ time,
  \item for a given key $t$, split the set $A$ into $A_{\le t}=\{a\in A : a\le t\}$ and $A_{>t}=\{a\in A : a > t\}$,
  which takes $\Oh(\log \min(|A_{\le t}|, |A_{> t}|))$ time.
\end{itemize}

\newcommand{\Tokens}{\mathit{Tokens}}
\begin{claim}
The $k$-complete family $N$ represented as a trie $T_N$ can be constructed
in $\Oh(|\F|2^k\binom{\log |\F|+k+1}{k+1})$ time
provided constant-time $\LCP$ queries for suffixes of strings $F\in \F$.
\end{claim}
\begin{proof}
To a tuple $(S,F,b) \in \F_P$ we assign a number of tokens:
$$\Tokens_P(S,F,b)=C (2^{b+1}-1)\tbinom{\log |\F_P|+b+1}{b+1}$$
where $C$ is a sufficiently large constant.
We shall inductively prove that
\FR($P,\F_P$) can be implemented in time
$$\sum_{(S,F,b)\in \F_P} \Tokens_P(S,F,b).$$

Before that, let us specify how the arguments to the procedure are specified.
The string $P$ is represented by the corresponding node of the constructed trie $T_N$;
we also explicitly store $|P|$ and $\#_{\$}(P)$.
The set $\F_P$ is stored in a finger search tree with tuples $(S,F,b)$
ordered by $S$. However, $S$ is not stored itself as it is uniquely specified as a suffix of $F$ of length $|F|-|P|$. Thus each element in the tree is stored in $\Oh(1)$ space. 

First, we process tuples $(S,F,b)$ with $S=\eps$. They are conveniently located at the beginning of $\F_P$.
We remove these tuples from $\F_P$ and store $F$ at the current node of $T_N$.
This simulates inserting $P$ to $N(F)$; we also store auxiliary values $\HD(P,F)=k-b$ and $\#_{\$}(P)$.

Next, we compute the length of longest common prefix $P'$ of non-empty strings $S$ with $(S,F,b)\in \F_P$.
For this, we make an $\LCP$ query for the smallest and the largest of these suffixes.
If the longest common prefix $P'$ is non-empty, we observe that $\F_{PP'}=\F_P$
(with the stored representation unchanged) and Algorithm \ref{alg:gen} does not explore any other branch.
Hence, we immediately call \FR($PP',\F_{PP'}$) which corresponds to creating a complete compacted edge of the resulting trie. This step takes $\Oh(1)$ time,
but it guarantees that \FR($PP',\F_{PP'}$) outputs or branches. Hence, this time gets amortized.

If $P'=\eps$, we partition $\F_P$ into at most $\sigma$
finger search trees $\F_{P,c}$ each storing tuples sharing the character $S[1]=c$,
and we identify the heavy character $h$ by choosing the largest $\F_{P,c}$.
For this, we iteratively split out the tree with the smallest unprocessed $S[1]$,
which takes time proportional to $\sum_{c \ne h} \log |\F_{P,c}|$. 

The sets $\F_{Pc}$ for $c\ne h$ already represented by $\F_{P,c}$ (note that the order does not change, and the tuples need not be altered since the ``budget'' $b$ remains the same and $S$ is stored implicitly).
Similarly, we can build $\F_{Ph}$ by inserting new tuples into $\F_{P,h}$.

Thus, we define $$\L_P := \{(S,F,b)\in \F_P : S\ne \eps \text{ and }S[1]\ne h\}$$
and insert to $\F_{Ph}$ and $\F_{P\$}$ tuples $(S[2\dd],F,b-1)$ for $(S,F,b)\in \L_P$ with $b>0$,
which takes $\Oh(\log |\F_P|)$ time per element.

In total, the processing time is $\Oh(1)$ for each element of $\L_P$ with $b=0$,
and $\Oh(\log |\F_P|)$ when $b>0$. Additionally, we may spend $\Oh(1)$ time for a tuple with $S=\eps$. Let us check that the difference in the number of tokens is sufficient to cover the running time of these operations.

The tuples with $S=\eps$ do not appear in future computations. Hence, we spend all their tokens on the computations related to them. It is indeed sufficient:
$$\Tokens_P(\eps,F,b) = C(2^{b+1}-1)\tbinom{\log |\F_P|+b+1}{b+1}\ge C\tbinom{\log |\F_P|+1}{1}=C(\log |\F_P|+1)\ge C.$$
We don't spend any time on tuples with $S[1]=h$, and number of tokens for such a tuple does not increase:
\begin{align*}
&\Tokens_P(S,F,b) - \Tokens_{Ph}(S,F,b)=\\
&C(2^{b+1}-1)\tbinom{\log |\F_P|+b+1}{b+1}- C(2^{b+1}-1)\tbinom{\log |\F_{Ph}|+b+1}{b+1}\ge 0.
\end{align*}
Finally, for a tuple with $S[1]\ne h$ (i.e., in $\L_p$) the difference in the number of tokens is
\begin{align*}
&\Tokens_P(S,F,b)-\Tokens_{Pc}(S',F,b)-\Tokens_{Ph}(S',F,b-1)\\
&-\Tokens_{P\$}(S',F,b-1)=\\
&C(2^{b+1}-1)\tbinom{\log |\F_P|+b+1}{b+1}-C(2^{b+1}-1)\tbinom{\log |\F_{Pc}|+b+1}{b+1}-C(2^{b}-1)\tbinom{\log |\F_{Ph}|+b}{b}\\
&-C(2^{b}-1)\tbinom{\log |\F_{P\$}|+b}{b}\\
&\ge C\tbinom{\log |\F_{P}|+b}{b}
\end{align*}
where $c=S[1]$ and $S'=S[2 \dd]$.
It is sufficient since we spend constant time for $b=0$ and $\Oh(\log|\F_P|)$ time for $b\ge 1$.

The claimed bound on the overall running time follows.
\end{proof}
\end{document}

%% file: __arxiv_preamble.tex
\usepackage{fullpage}
\usepackage{amsthm,amssymb,amsmath}  
\usepackage{xspace,enumerate}
\usepackage[hidelinks]{hyperref}
\usepackage[capitalise]{cleveref}
\usepackage[utf8]{inputenc}
\usepackage{thmtools}
\usepackage{thm-restate}
\usepackage{authblk}

  \theoremstyle{plain}
  \newtheorem{theorem}{Theorem}
  \newtheorem{lemma}{Lemma}  
  \newtheorem{corollary}[theorem]{Corollary}  
  \newtheorem{fact}{Fact}
  \newtheorem{observation}{Observation}
  \theoremstyle{definition}
  \newtheorem{definition}{Definition}
  \newtheorem{problem}{Problem}
  \newtheorem{example}{Example}
  \newtheorem{remark}[definition]{Remark}
 
  \newtheorem*{claim}{Claim}
 
\title{Linear-Time Algorithm for Long LCF with $k$ Mismatches}
\author[1]{Panagiotis Charalampopoulos}
\author[1]{Maxime Crochemore}
\author[1]{Costas S. Iliopoulos}
\author[2]{Tomasz Kociumaka}
\author[1]{Solon P. Pissis}
\author[2]{Jakub Radoszewski}
\author[2]{Wojciech Rytter}
\author[2]{Tomasz Wale\'n}

\affil[1]{Department of Informatics, King's College London, London, UK\\
\texttt{[panagiotis.charalampopoulos,maxime.crochemore,}\\
\texttt{costas.iliopoulos,solon.pissis]@kcl.ac.uk}}
\affil[2]{Faculty~of Mathematics, Informatics and Mechanics, University of Warsaw, Warsaw, Poland\\
\texttt{[kociumaka,jrad,rytter,walen]@mimuw.edu.pl}}

\date{\vspace{-5ex}}


%% file: _fig_diff_cover_example.tex
\begin{center}
\begin{tikzpicture}[scale=0.5]
\tikzstyle{s} = [draw, circle, fill=black, minimum size = 3pt, inner sep = 0 pt, color=black]
\draw (1,0)--(24,0);
 
\foreach \i in {1,...,24} {
  \draw (\i,-0.1)--+(0,0.2);
  \node at (\i, 0.2) [above] {\small \i};
}
\foreach \i in {2,3,5,8,9,11,14,15,17,20,21,23} {
  \node [s] at (\i,0) {};
}

\foreach \i/\l in {3/5, 10/5} {
  \draw [dotted] (\i,0)--+(0,-1.1);
  \draw [dotted] (\i+\l,0)--+(0,-1.1);
  \draw [-latex] (\i,-.8)--node[midway,below] {$h(3,10)=5$} +(\l,0);
}
\end{tikzpicture}
\end{center}

%% file: _fig_tree_problem_ex1.tex
\begin{tikzpicture}
\tikzset{
  sibling distance=0.6em,
  level distance=2em,
  edge from parent/.style={draw, edge from parent path={(\tikzparentnode) -- (\tikzchildnode)}},
  inner/.style = {circle, fill=black, inner sep=0pt, minimum size=2pt, draw},
  leafR/.style = {shape=rectangle, draw, align=center, color=red!80!black},
  leafB/.style = {shape=circle, draw, align=center, color=blue!80!black, inner sep=2pt}
}
\matrix{
\tikzset{frontier/.style={distance from root=7em}}
\Tree [
  . \node[inner] {} node[above, xshift=-3em, yshift=1em] {$T_1$};
  [
    .\node[inner] {};
    [ 
      .\node[inner] {};
      [ \node[leafB] {$1$}; ]
      [ \node[leafB] {$2$}; ]
    ]
    [ 
      .\node[inner] {};
      [ \node[leafR] {$1$}; ]
      [ \node[leafB] {$3$}; ]
    ]
  ]
  [
    .\node[inner] {}  node[above, right, yshift=0.2em] {$v_1$};
    [ 
      .\node[inner] {};
      [ \node[leafR] {$2$}; ]
      [ \node[leafR] {$3$}; ]
    ]
    [ 
      .\node[inner] {};
      [ \node[leafB] {$4$}; ]
      [ \node[leafR] {$4$}; ]
    ]
  ]
]
&
\draw[white] (0,0)--+(2em,0);
&
\tikzset{frontier/.style={distance from root=7em}}
\Tree [
  . \node[inner] {} node[above, xshift=-3em, yshift=1em] {$T_1$};
  [
    .\node[inner] {};
    [ 
      .\node[inner] {};
      [ \node[leafB] {$1$}; ]
      [ \node[leafB] {$3$}; ]
    ]
    [ 
      .\node[inner] {} node[above, right, yshift=0.2em] {$v_2$};
      [ \node[leafB] {$4$}; ]
      [ \node[leafR] {$2$}; ]
    ]
  ]
  [
    .\node[inner] {};
    [ 
      .\node[inner] {};
      [ \node[leafR] {$1$}; ]
      [ \node[leafR] {$3$}; ]
    ]
    [ 
      .\node[inner] {};
      [ \node[leafR] {$4$}; ]
      [ \node[leafB] {$2$}; ]
    ]
  ]
]
\\
};
\end{tikzpicture}